\documentclass{l4dc2026}

\title[Learning Control Barrier Functions]{Learning Control Barrier Functions with Deterministic Safety Guarantees }

\usepackage{times}
\usepackage{acronym}
\usepackage{mathtools}
\usepackage{comment}
\usepackage{wrapfig}
\usepackage{soul}
\usepackage[T1]{fontenc}

\newcommand{\defeq}{\vcentcolon=}

\def\IR{\mathbb R}					

\newtheorem{cor}{Corollary}

\newtheorem{problem}{Problem}

\newtheorem{assumption}{Assumption}

\newcommand{\Tcal}{\mathcal{T}}

\newcommand{\Scal}{\mathcal{S}} 
\newcommand{\I}{\mathbf{I}} 
\newcommand{\mbf}[1]{\mathbf{#1}} 
\newcommand{\bmat}[1]{\begin{bmatrix} #1 \end{bmatrix}} 

\newcommand{\A}{\mathbf{A}} 
\newcommand{\B}{\mathbf{B}} 
\newcommand{\x}{\mathbf{x}} 
\newcommand{\bu}{\mathbf{u}} 
\newcommand{\y}{\mathbf{y}} 
\newcommand{\w}{\mathbf{w}} 

\newcommand{\M}{\mbf{M}}

\newcommand{\W}{\mbf{W}}

\newcommand{\X}{\mbf{X}}

\newcommand{\norm}[1]{\left\lVert#1\right\rVert}
\newcommand{\absVal}[1]{\left\lvert#1\right\rvert}

\newcommand{\Ucal}{\mathcal{U}}
\newcommand{\Xcal}{\mathcal{X}}

\DeclareMathOperator*{\argmin}{arg\,min}

\newcommand{\Kcal}{\mathcal{K}}
\newcommand{\Dcal}{\mathcal{D}}

\newcommand{\sumLam}{\sum_{j=0}^n\lambda_j}
\newcommand{\Wb}{\bar{W}}

\newcommand{\vecdim}[1]{\in \mathbb{R}^{#1}}
\newcommand{\mapX}[2]{#1\rightarrow#2}
\newcommand{\map}[2]{ \mapX{\mathbb{R}^{#1}} {\mathbb{R}^{#2}} }

\acrodef{dl}[{DL}]{deep learning}
\acrodef{rl}[{RL}]{reinforcement learning}
\acrodef{nn}[{NN}]{neural network}
\acrodef{dnn}[{DNN}]{deep neural network}
\acrodef{tdl}[{TDL}]{temporal difference learning}
\acrodef{pid}[{PID}]{proportional–integral–derivative}
\acrodef{us}[{US}]{Ultrasound}
\acrodef{mse}[{MSE}]{mean squared error}
\acrodef{sgd}[{SGD}]{stochastic gradient descent}
\acrodef{ico}[{ICO}]{iterative convex overbounding}
\acrodef{lmi}[{LMI}]{linear matrix inequality}
\acrodef{mjls}[{MJLS}]{Markov jump linear system}
\acrodef{io}[{IO}]{input-output}
\acrodef{iqc}[{IQC}]{integral quadratic constraints}
\acrodef{cnn}[{CNN}]{convolutional neural network}
\acrodef{il}[{IL}]{Imitation learning}
\acrodef{mpc}[{MPC}]{model predictive control}
\acrodef{sdp}[{SDP}]{semi-definite programming}
\acrodef{relu}[{ReLU}]{rectified linear unit}
\acrodef{us}[US]{ultrasound}
\acrodef{mdp}[MDP]{Markov Decision Process}
\acrodef{iid}[iid]{identical and independently distributed random variable}
\acrodef{pid}[PID]{Proportional Integral Derivative}

\acrodef{lqr}[LQR]{linear-quadratic regulator}
\acrodef{lqr}[LQR]{linear-quadratic regulator}
\acrodef{fvi}[FVI]{fitted value iteration}
\acrodef{ols}[OLS]{ordinary least squares}

\acrodef{cpa}[CPA]{continuous piecewise affine}
\acrodef{lti}[LTI]{linear time invariant}

\acrodef{cbf}[CBF]{control barrier function}
\acrodef{bf}[BF]{barrier function}
\acrodef{admm}[ADMM]{alternating direction method of multipliers}

\acrodef{ci}[CI]{control invariant}
\acrodef{hjb}[HJB]{Hamilton Jacobi Bellman}
\acrodef{smt}[SMT]{satisfiability modulo theory}
\acrodef{sos}[SOS]{sum of squares}

\coltauthor{\Name{Amy K. Strong} \Email{aks121@duke.edu}\and
  \Name{Leila J. Bridgeman} \Email{ljb48@duke.edu}\\
  \addr Department of Mechanical Engineering and Materials Science at Duke University,
Durham, NC, 27708, USA
  \AND
 \Name{Ali Kashani} \Email{fww9ba@virginia.edu}\\
 \addr School of Data Science at the
University of Virginia, Charlottesville, VA, 22903, USA%
\AND
 \Name{Claus Danielson} \Email{cdanielson@unm.edu}\\
 \addr Department of Mechanical Engineering at
the University of New Mexico, Albuquerque, NM, 87131, USA%
}

\begin{document}

\maketitle

\begin{abstract} 
    \Acp{bf} characterize safe sets of dynamical systems, where hard constraints are never violated as the system evolves over time. 
    Computing a valid safe set and \ac{bf} for a nonlinear (and potentially unmodeled), non-autonomous dynamical system is a difficult task. This work explores the design of \acp{bf} using data to obtain safe sets with deterministic assurances of control invariance. 
    We leverage ReLU \acp{nn} to create \ac{cpa} \acp{bf} with deterministic safety guarantees for Lipschitz continuous, discrete-time dynamical system using sampled one-step trajectories.
    The \ac{cpa} structure admits a novel classifier term to create a relaxed \ac{bf} condition and construction via a data driven constrained optimization.
    We use \ac{ico} to solve this nonconvex optimization problem through a series of convex optimization steps. We then demonstrate our method's efficacy on two-dimensional autonomous and non-autonomous dynamical systems.
\end{abstract}


\section{Introduction}


An invariant set is a region of the state space where a dynamical system remains in perpetuity -- either autonomously via the dynamics alone or non-autonomously through admissible control inputs. When an invariant set is a subset of the state constraint set ($\Scal\subseteq\Xcal$), it is a ``safe set", i.e., a region where constraints will never be violated.  
While methods exist for data driven safe set determination, they often approximate the safe set \citep{korda2020computing} or provide probabilistic guarantees of invariance \citep{wang2020scenario}. 

A \ac{bf} is a scalar function of state variables whose zero superlevel or sublevel set denotes the safe set \citep{ames2016control}. The system must satisfy a \ac{bf} condition that modulates system evolution over time to ensure the state trajectories always remain within the safe set, i.e. ensure invariance. The \ac{bf} condition takes the form of an inequality constraint that must be held for all states in the state constraint set, $(\forall \x \in\Xcal)$. Synthesizing a \ac{bf} for a given system and its constraints is often a difficult, system-dependent process.

As the control community moves towards developing model-free designs, it is crucial we develop model-free methods to assure safety. 
This work explores the design of \acp{bf} using data to obtain safe sets with deterministic assurances of control invariance.

\subsection{Related Work}
Often, in data driven \ac{bf} synthesis, the safe set is assumed to be fully \citep{robey2020learning,dawson2023safe} or partially known \citep{qin2021learning,anand2023formally,nejati2023data}. However, in practice, the safe sets are often unavailable. Recent data-driven literature is aligned to find both the safe set and corresponding \ac{bf} \citep{kashani2025data,liu2025computing,so2024train}.
The \ac{bf} is typically modeled as a parameterized function, such as linear function \citep{saveriano2019learning}, polynomial \citep{nejati2023data,lefringhausen2025barrier}, or \ac{nn} \citep{jin2020neural,qin2022sablas,anand2023formally,mathiesen2023safety,edwards2024fossil}, where the aim is to learn the parameters that enforce the \ac{bf} condition. 

It is often difficult to ensure deterministic safety guarantees in data driven settings, i.e., adherence of \ac{bf} condition for the unseen states. Specifically, the \ac{bf} condition for a \ac{nn} \ac{bf} is often imposed as a soft constraint in the loss function. Guarantees of enforcing the \ac{bf} condition are shown probabilistically \citep{qin2021learning}. Alternatively, verification of \ac{bf} is performed after synthesis~\citep{edwards2024fossil}. \Ac{smt} methods can verify the \ac{nn} adheres to the \ac{bf} conditions using an iterative learner verifier architecture \citep{edwards2024fossil,peruffo2021automated}. Arguably, the most popular one-shot method is $\epsilon-$net coverage \citep{robey2020learning,anand2023formally}, where the density of samples and the Lipschitz continuity of the \ac{bf} and dynamical system are leveraged for deterministic guarantees of safety.

\subsection{Contributions}
Our work extends the field of \ac{bf} synthesis by learning both a \ac{bf} and its corresponding safe set for any Lipschitz continuous, discrete-time dynamical system from one-step state transition samples. We first develop a necessary and sufficient \ac{bf} condition for discontinuous \acp{bf} that uses a discontinuous classifier function to determine the invariance condition applied to each state within the state constraint set. We then leverage this condition to develop a nonconvex optimization program that simultaneously learns a \acf{relu} \acf{nn} barrier function and its classifier function. We implement the non-convex optimization method, \acf{ico} \citep{de2000convexifying,warner2017iterative,sebe2018sequential}, to efficiently find the \ac{bf} -- guaranteeing convergence of the problem to some local minima. 
Finally, we develop a method to prompt additional sampling, if possible, to maximize the volume of the safe set while maintaining the feasibility of the learning problem.


A key advancement of our work is the way the \ac{relu} \ac{nn} \ac{bf} are learned. Enforcing the \ac{bf} condition would require constrained gradient descent techniques through iterative approaches, which may render the problem infeasible. However, a \ac{relu} \ac{nn} is a \ac{cpa} function \citep{arora2016understanding}. Instead of learning and verifying the \ac{nn} weights layer-by-layer, we directly learn the \acf{cpa} function through an optimization problem. This change in perspective allows us to benefit from the simplicity of \ac{relu} \acp{nn} while verifying \ac{bf} conditions. Further, the local structure of the \ac{cpa} function can be leveraged to require less data density than $\epsilon$-net methods.

\subsection*{Notation and Preliminaries}
The notation $\mathbb{Z}_a^b$ ($\mathbb{Z}_{\bar{a}}^{\bar{b}}$) denotes the set of integers between $a$ and $b$ inclusive (exclusive), and $\mathbb{R}_a^b$ ($\mathbb{R}_{\bar{a}}^{\bar{b}}$) denotes the set of real numbers between $a$ and $b$ inclusive (exclusive).
The interior, boundary, and closure of the set $\Omega \subset \mathbb{R}^n$ are denoted as $\Omega^\circ$, $\delta \Omega$ and $\bar{\Omega},$ respectively. The notation $\mathfrak{R}^n$ denotes the set of all compact subsets satisfying i) $\Omega^\circ$ is connected and ii) $\Omega = \bar{\Omega}^\circ.$ 

Scalars, vectors, and matrices are denoted as $x,$ $\x,$ and $\X$, respectively. The mapping $g:\Xcal \rightarrow \mathcal{Y}$ between two metric spaces is said to be a Lipschitz mapping if there is some Lipschitz constant, $L > 0$, such that $\norm{g(\mathbf{p}) - g(\mathbf{q})} \leq L \norm{\mathbf{p} - \mathbf{q}}$ for all points $\mathbf{p}$ and $\mathbf{q}$ in $\Xcal$ \citep{fitzpatrick2009advanced}. Let $g^k$ indicate that a mapping has been applied $k$ times, where $k\in\mathbb{Z}_{0}^{\infty}$. Let $\varphi:\IR_0^\infty \rightarrow \IR_0^\infty$ be a class $\Kcal$ function if it is continuous, strictly increasing, and $\varphi(0) = 0.$

A set $\Scal$ is \ac{ci} under the dynamics of \eqref{eq:dynSys} if, $ \forall\x_0 \in \Scal$, there exists $\bu_k \in \Ucal$ such that $g(\x_k,\bu_k) \in \Scal$ for $k \in\mathbb{Z}_{0}^{\infty}$ \citep{borrelli2017predictive}. Here, a safe set refers to a constraint admissible \ac{ci} set, i.e., a \ac{ci} set that is a subset of the state constraint set ($\Scal \subseteq \Xcal).$

\section{Problem Statement and Preliminaries}

Consider the discrete dynamical system,
\begin{equation}\label{eq:dynSys}
    \x^+ = g(\x,\bu), \quad \quad \x \in \Xcal, \bu \in \Ucal,
\end{equation}
where $g: \IR^n\times \IR^m\rightarrow \IR^n,$ $\Xcal\subset \mathfrak{R}^n$ and $\Ucal\subset \mathfrak{R}^m$ are the state and input constraint sets, respectively. This paper aims to compute a \ac{bf}, $W:\Xcal \rightarrow \mathbb{R}$ and corresponding safe set, $\Scal \subseteq \Xcal$, from a sampled data set of $NM$ one-step mappings produced by \eqref{eq:dynSys}, i.e. $\Dcal = \{\x_{z},\{\bu_{z,k}, \x_{z,k}^+\}_{k=1}^M\}_{z=1}^N,$ where a state, $\x_z$, may be repeatedly sampled with multiple inputs $\{\bu_{z,k}\}_{k=1}^M$. One-step trajectories allow for flexible sampling strategies; $\Dcal$ can be stochastically or deterministically sampled and can easily incorporate multi-step trajectories.

We use Lipschitz continuity to extrapolate system behavior from data points to nearby regions without a model. Many methods have been proposed to estimate the Lipschitz constant from data \citep{huang2023sample,nejati2023data,wood1996estimation,strongin1973convergence,sergeyev1995information}, making this a reasonable assumption in a data-driven context.
\begin{assumption}\label{assum:Lipschitz}
    The system \eqref{eq:dynSys} is locally Lipschitz continuous on $\Xcal \times \Ucal$. Moreover, there exists a known (but not necessarily tight) Lipschitz constant $L > 0$ with respect to the Euclidean norm $\norm{\cdot}_2$.
\end{assumption}

\subsection{Barrier Functions and Safe Sets}

\Acp{bf} are scalar functions that denote a safe set of a system via a sublevel set. In \citep[Theorem 1]{ahmadi2019safe}, necessary and sufficient conditions for invariance were found for continuous \acp{bf}.

\begin{definition}\cite[Definition 2]{ahmadi2019safe}\label{def:discreteBF_OG}
    For \eqref{eq:dynSys}, the continuous function $W:\IR^n\rightarrow\IR$ is a discrete time barrier function for the set $\Scal \defeq \{\x\in\Xcal\mid W(\x) \geq 0\} \subseteq \Xcal \subset \IR^n$ if there exists $k:\Xcal\rightarrow\Ucal$ and $\gamma \in \mathrm{class}\ \Kcal$ and  satisfying $\gamma(r) < r$ for all $r>0$ such that if $\bu(x)=k(x)$ then
     \begin{flalign}\label{eq:ogIncrease}
        W(\x^+) - W(\x) \geq -\gamma(W(\x)), \quad \forall \x \in \Xcal.
    \end{flalign}
\end{definition}

In the following lemma, we extend these necessary and sufficient conditions of invariance to discontinuous $W$ and $\gamma$. Further, we consider non-autonomous systems by considering viable $\bu\in\Ucal$ that produce $\x^+$ where \eqref{eq:ogIncrease} holds.

\begin{lemma}\label{lem:invSetCond}
Consider the \eqref{eq:dynSys} and let Assumption \ref{assum:Lipschitz} hold. Define the function $\gamma:\mapX{\IR \times \IR^n}{\IR}$, where $\gamma(0,\x) = 0$ and $\langle \gamma (y,\x), y \rangle \geq 0$ for all $y\vecdim{}, \x\vecdim{n}.$ For $W:\mapX{\Xcal}{\IR}$, the set $\Scal \defeq \{\x \in \Xcal \mid W(\x) \leq 0\}\subseteq \Xcal$ is a safe set if and only if there exists $\bu \in \Ucal$ and $\gamma$ such that
\begin{flalign}\label{eq:decreaseBF}
    W(g(\x,\bu)) - \gamma(W(\x),\x) \leq 0, \quad \forall\x \in \Xcal.
\end{flalign}
\end{lemma}
\begin{proof}
    First, we show sufficiency. Because $\langle \gamma (y,\x), y \rangle \geq 0,$ $\gamma$ always has the same sign as its argument. Let \eqref{eq:decreaseBF} hold for all $\x \in\Xcal$ for $\bu = 0$ (autonomous systems) or for a viable $\bu \in\Ucal$ for each $\x \in\Xcal$ (non-autonomous systems). If $W(\x) \leq 0$, then $\gamma(W(\x),\x) \leq 0$ and $W(g(\x,\bu)) \leq 0.$ By definition, $\Scal$ is a safe set.
    
    To show necessity, note that if $\Scal  \defeq \{\x \in \Xcal \mid W(\x) \leq 0\}$ is a safe set, then the proximal function of $W$ satisfies the required conditions. To see why, let
    \begin{flalign}\label{eq:proxGamma}
    \gamma(W(\x),\x)=\Pi(W(\x)) = \begin{cases}
        0 & \text{if } W(\x) \leq 0 \\
        \infty & \text{if } W(\x) > 0.
        \end{cases}
    \end{flalign}
    Then $W(\x)=0$ implies $\x\in \Scal$, so $\gamma(0,\x)=0$. Further, $\langle \gamma (y,\x), y \rangle = \langle \Pi (y), y \rangle\geq 0$ because if $y\leq 0 $, then $\langle \Pi (y), y \rangle=\langle 0, y \rangle=0$, and if $y>0$, then $\langle \Pi (y), y \rangle=\langle \infty, y \rangle=\infty>0$. 
    Finally, to show that \eqref{eq:decreaseBF} holds, consider first $x\in\Scal\subseteq\Xcal.$ In this case, there is a $\bu\in\Ucal$ such that $\x,g(\x,\bu)\in\Scal,$ so $W(g(\x,\bu))\leq 0=\gamma(W(\x),\x)$. If $\x\in\Xcal\setminus\Scal,$ then $W(g(\x,\bu))\leq \infty=\gamma(W(\x),\x)$.
\end{proof}

Essentially, $\gamma$ is a sign classifier. The least conservative choice of $\gamma$ would be the proximal function of $W$, \eqref{eq:proxGamma}. 
However, the proximal function is not an appropriate choice for data driven scenarios, wherein we are trying to learn $W$ and $\Scal.$ The crux of any data driven problem is extrapolating information from a data point to nearby unsampled regions. If the proximal function is used, then any data point at the boundary of $\Scal$ would have an infinite difference between itself and nearby states. Thus, the challenge of this paper is to determine a useful $\gamma$ to synthesize $W$ using a finite data set. The function $\gamma$ must classify a sampled point and the region around that sampled point.

\subsection{CPA Functions}

In this work, we assume the \ac{bf} has a predetermined functional form, a \ac{cpa} function. This is equivalent to the output of a \ac{relu} \ac{nn} \citep{arora2016understanding, he2018relu}, but the \ac{cpa} representation will admit error bounds via Lemma \ref{lemma:lipBound} below that are crucial to our guarantees. 

\begin{definition} Affine independence
\citep{giesl2014revised}:
    A collection of $m$ vectors, \\$\{\x_0,\x_1, \hdots , \x_m\} \subset \mathbb{R}^n$, is affinely independent if $\x_1-\x_0, \hdots, \x_m - \x_0$ are linearly independent.  
\end{definition}
\begin{definition} $n$ - simplex \citep{giesl2014revised}:
    A simplex, $\sigma$, is defined as the convex hull of $n+1$ affinely independent vectors, $co\{\x_j\}_{j=0}^n$, where each vector, $\x_{j} \in \mathbb{R}^n$, is a vertex.
\end{definition}
\begin{definition} Triangulation \citep{giesl2014revised}:
    Let  $\Tcal {=} \{\sigma_i\}_{i=1}^{m_{\Tcal}} {\in} \mathfrak{R}^{n}$ represent a finite collection of $m_{\Tcal}$ simplices, where the intersection of any two simplices is a face or an empty set.
\end{definition}
Let $\Tcal = \{\sigma_i\}_{i=1}^{m_{\Tcal}},$ 
where $\sigma_i =\text{co}(\{\x_{i,j}\}_{j=0}^n).$ All vertices of the triangulation $\Tcal$ are denoted as $\mathbb{E}_{\Tcal}.$

\begin{lemma}\label{lemma_gradW} (\cite[Remark 9]{giesl2014revised}):
    Consider $\Tcal {=} \{\sigma_i\}_{i=1}^{m_{\Tcal}},$ where $\sigma_i {=} \text{co}(\{\x_{i,j}\}_{j=0}^n)$, and a set $\mathbf{W} {=} \{W_{\x}\}_{\x \in \mathbb{E_{\Tcal}}}{\subset}\mathbb{R},$ where $W(\x) {=} W_\x,$ $\forall \x\in\mathbb{E}_{\Tcal}.$ For simplex $\sigma_i$, let $\X_i \vecdim{n \times n}$ be a matrix that has $\x_{i,j} {-} \x_{i,0}$ as its $j$-th row and $\bar{W}_i\vecdim{n}$ be a vector that has $W_{\x_{i,j}} {- }W_{\x_{i,0}},$ as its $j$-th element. The function $W(\x) {=} \x^\top \X_i^{-1}\bar{W}_i$ is the unique \ac{cpa} interpolation of $\W$ on $\Tcal$ for $\x \in \sigma_i.$
\end{lemma}

Lemma \ref{lemma:lipBound} uses system dynamics to develop an error bound that characterizes the difference between the system evaluated at the vertices of a simplex and evaluated within the simplex.

\begin{lemma}  (\cite[Proposition 4.1]{baier2012linear}): \label{lemma:lipBound}
    Let $\{ \x_j\}_{j=0}^k\subset \IR^n$ be affinely independent vectors. Define $\sigma \defeq \text{co}(\{ \x_j\}_{j=0}^k),$ $c = \max_{\x,\y \in \sigma}\norm{\x - \y},$ and consider a convex combination $\sumLam \x_j \in \sigma.$ If $g: \Omega \rightarrow \mathbb{R}$ is Lipschitz with constant $L,$ then
    \begin{equation}
        \absVal{g\left(\sumLam \x_j\right) - \sumLam g(\x_j)} \leq Lc.
    \end{equation}
\end{lemma}

\section{CPA Barrier Function}

With the established \ac{bf} condition in Lemma \ref{lem:invSetCond}, we now determine a specific functional for both the \ac{bf} $W$ and the decision function, $\gamma$, that is amenable to data driven learning using the data set of one-step trajectories, $\Dcal = \{\x_{z},\{\bu_{z,k}, \x_{z,k}^+\}_{k=1}^M\}_{z=1}^N$. First, the \ac{bf} $W$ is characterized as a \ac{cpa} function over a triangulation,$\Tcal \subseteq \Xcal$, where the data samples $\{\x_z\}_{z=0}^N$ form the vertices of the triangulation. Here, $W$ is affine on each simplex of the triangulation. In Theorem \ref{thm:invSetConditions}, the unique local structure of $W$ is used to ensure the \ac{bf} condition in the Lemma \ref{lem:invSetCond}  is enforced.

Note that $W$ is both a \ac{cpa} function and a \ac{relu} \ac{nn}. The advantage of directly learning the \ac{cpa} function is that we can impose hard constraints on the function \citep{giesl2014computation}. In contrast, learning a \ac{relu} \ac{nn} relies on either soft constraints imposed in a loss function \citep{anand2023formally} or verification of constraint adherence after training \citep{zhang2024seev}. 

Next, we define $\gamma$ as a piecewise linear function, $\gamma(y,\x) = \gamma_iy,$ where $\gamma_i\geq0$ is constant on each simplex of the triangulation, $\sigma_i \in \Tcal.$ Condition \eqref{eq:decreaseBF} from Lemma \ref{lem:invSetCond} becomes
\begin{flalign}\label{eq:decreaseAlpha}
    \min_{k\in\mathbb{Z}_1^M} W(\x_{i,j,k}^+) - \gamma_i W(\x_{i,j}) \leq 0,
\end{flalign}
where $\x_{i,j}$ defines a vertex in $\Tcal$ ($i\in\mathbb{Z}_1^{m_\Tcal}$, $j \in \mathbb{Z}_0^n$). Because there are many potential successor states $\{\x_{i,j,k}^+\}_{k=1}^M$ based on the input samples associated with each state, $\{\bu_{i,j,k}\}_{k=1}^M$, the successor state that minimizes \eqref{eq:decreaseAlpha} is chosen. A single input sample per sampled state is allowable as well.

Thus, on each simplex $\sigma_i\in\Tcal$, the nearby data samples of $\{\x_{i,j}\}_{j=0}^n$ that make up the vertices of $\sigma_i$ are used to inform the \ac{bf} condition applied to all $\x \in\sigma_i.$ Condition \eqref{eq:decreaseAlpha} provides a conservative approximation of the proximal function, \eqref{eq:proxGamma}. The selection of $\gamma_i$ is also inherently tied to the sign of $W(\x_{i,j})$ at the vertex points, $\{\x_{i,j}\}_{j=0}^n$. If $W(\x_{i,j})$ is negative for all $j\in\mathbb{Z}_0^n$, then $\gamma \in \mathbb{R}_0^1$ provides the least conservative \ac{bf} condition. If $W(\x_{i,j})$ is positive for all $j\in\mathbb{Z}_0^n$, then $\gamma \in \mathbb{R}_1^{\underline{\infty}}$ allows for the system to have less constrained evolution over time. For simplices with a mix of positive and negative $W(\x_{i,j})$ at the vertices, whichever $\gamma \in \mathbb{R}_0^{\underline{\infty}}$ that provides a feasible condition is used.

In Theorem \ref{thm:invSetConditions}, the specific forms of $W$ and $\gamma$ are used to develop conditions for a \ac{cpa} barrier function. Later, an optimization problem is developed to learn both $W$ and $\gamma$, while maximizing the volume of the safe set.

\begin{theorem}\label{thm:invSetConditions}
    Consider \eqref{eq:dynSys}, and let Assumption \ref{assum:Lipschitz} hold. Let $\Ucal$ be a convex set. Let the data set $\Dcal = \{\x_{z},\{\bu_{z,k}, \x_{z,k}^+\}_{k=1}^M\}_{z=1}^N$ be $NM$ one step trajectories generated by \eqref{eq:dynSys} in $(\Xcal,\Ucal)$. Let $\Tcal = \{\sigma_i\}_{i=1}^{m_{\Tcal}} \subseteq \Xcal$ be a triangulation over $\{\x_z\}_{z=0}^N$.  Let $\gamma = \{\gamma_i\}_{i=1}^{m_\Tcal} \subset \IR.$ Let $\W = \{W_{\x}\}_{\x \in \mathbb{E}_{\Tcal}}\subset \mathbb{R},$ where $W:\map{n}{}$ is the \ac{cpa} interpolation of $\W$ on $\Tcal$ for any $\x\in\sigma_i$. Define the $\bar{W}$ as
    \begin{flalign}\label{eq:Wbar}
        & \bar{W} = \begin{cases}
             W(\x) \quad &\forall \x \in \Tcal \\
             \epsilon \quad &\forall \x \notin \Tcal,
        \end{cases}
    \end{flalign}
    where $\epsilon \in \mathbb{R}_{0}^{\underline{\infty}}$.
    Let $\bar{W}$ satisfy
    \begin{subequations}
    \begin{align}
        \bar{W}(\x) \geq \epsilon, \quad \quad &\forall \x^+ \notin \Tcal, \label{eq:wLabel}\\
        W(\x) \geq -\rho, \quad \quad &\forall \x\in\mathbb{E}_\Tcal, \\
        |\nabla \bar{W}_i| \leq b, \quad \quad &\forall i \in \mathbb{Z}_1^{m_{\Tcal}}, \label{eq:gradW} \\
        \min_{k\in\mathbb{Z}_1^M}\bigl(\bar{W}(\x_{i,j,k}^+) - \gamma_i W(\x_{i,j}) + b L_g c_i\bigr) < 0,  \quad \quad &\forall i \in \mathbb{Z}_1^{m_{\Tcal}}, \forall j \in \mathbb{Z}_0^n, \label{eq:decreaseW} \\
        \gamma_i \geq 0, \quad \quad &\forall i \in \mathbb{Z}_1^{m_\Tcal} \label{eq:boundAlpha}
    \end{align}
    \end{subequations}
    where $L_g$ is the Lipschitz constant of \eqref{eq:dynSys} and $\nabla W(\x)=\nabla W_i$ for $\x\in\sigma_i$. Moreover, $\rho \in \mathbb{R}_{0}^{\underline{\infty}}$ and
    \begin{flalign}\label{eq:c}
        c_i \geq \max_{r,s\in \mathbb{Z}_0^n} \norm{\bmat{\x_{i,r}\\ \bu_{i,r}}-\bmat{\x_{i,s}\\\bu_{i,s}}}_2,
    \end{flalign}
    where $\bu_{i,r}$ and $\bu_{i,s}$ depend on the minimization in \eqref{eq:decreaseW}.
    Then, the sublevel set $\Scal \defeq \{\x \in \Xcal \mid \bar{W}(\x) \leq 0\}$ is a control invariant set. Further, because $\Scal \subseteq \Xcal$, $\Scal$ is a safe set.
\end{theorem}

\begin{proof}
     Consider \eqref{eq:decreaseBF} and let $\gamma(W(\x)) = \gamma_i W(\x),$ where $\gamma(W(\x)) = \gamma_i W(\x)$ when $\x\in\sigma_i.$ By Constraint \eqref{eq:boundAlpha}, $\gamma_i W(\x)$ satisfies the condition $\langle \gamma_i W(\x), W(\x) \rangle \geq 0$ $\forall W(\x) \in \Xcal$ and $\forall \gamma_i \in \Gamma.$ Therefore, the condition $\bar{W}(g(\x,\bu)) - \gamma_i W(\x) \leq 0$ is a specific instance of \eqref{eq:decreaseBF}. Constraint \eqref{eq:decreaseW} then assures that this instance of \eqref{eq:decreaseBF} holds for all $\x \in \Tcal,$ which proven below. 
     
     By definition, any point $\x \in \sigma_i$ is some convex combination of the vertex points, i.e. $\x = \sumLam \x_{i,j}.$ Therefore, for any $\x \in\sigma_i$ the barrier condition can be expressed as 
     \begin{flalign*}
           & \Wb(g(\sumLam \x_{i,j},\bu^*)) - \sum_{j=0}^n\lambda_j \gamma_i\Wb(\x_{i,j}) \leq 0,
     \end{flalign*}
     where $\sum_{j=0}^n\lambda_j \gamma_i\Wb(\x_{i,j}) =  \gamma_i\Wb(\sum_{j=0}^n\lambda_j\x_{i,j})$ due to linearity of $\gamma_iW(\x)$ on a simplex. Here, $\bu^*$ is some input which ensures the barrier condition holds. Consider $\sumLam \Wb(g(\x_{i,j},\bu_{i,j,k}))$, where $\bu_{i,j,k}$ some sampled input.
     By adding and subtracting $\sumLam \Wb(g(\x_{i,j},\bu_{i,j,k}))$, the barrier condition is also equivalent to
     \begin{flalign*}
         & \sumLam \Wb(g(\x_{i,j},\bu_{i,j,k})) + \Wb(g(\sumLam \x_{i,j},\bu^*)) - \sumLam \Wb(g(\x_{i,j},\bu_{i,j,k})) - \sum_{j=0}^n\lambda_j \gamma_i\Wb(\x_{i,j}) \leq 0.
     \end{flalign*}
     The expression is bounded above by
    \begin{flalign*}
        & \sumLam \Wb(g(\x_{i,j},\bu_{i,j,k})) +  \delta - \sum_{j=0}^n\lambda_j \gamma_i W(\x_{i,j}) \leq 0,
    \end{flalign*}
    where $\delta =  \absVal{\Wb(g(\sumLam \x_{i,j},\bu^*)) {-} \sumLam \Wb(g(\x_{i,j},\bu_{i,j,k}))}.$

    One valid candidate input is $\bu^*=\sumLam \bu_{i,j,k}$, which lies in the convex hull of $\bu_{i,j,k}\in\Ucal$. 
    Therefore, 
    $\delta=\absVal{\Wb(g(\sumLam \x_{i,j},\sumLam \bu_{i,j,k})) - \sumLam \Wb(g(\x_{i,j},\bu_{i,j,k}))}.$
    By \citep{giesl2012construction} and \citep{strong2024data}, this expression is bounded above by \\$b \norm{g(\sumLam \x_{i,j},\sumLam \bu_{i,j}) - \sumLam g(\x_{i,j},\bu_{i,j,k})}_2$. Using Lemma \ref{lemma:lipBound}, \\$b \norm{g(\sumLam \x_{i,j},\sumLam \bu_{i,j,k}) - \sumLam g(\x_{i,j},\bu_{i,j,k})}_2 \leq b L c_i$.

    Therefore for any $\x\in\sigma_i$ and for the feasible input $\bu^* = \sumLam \bu_{i,j,k}$, $\bar{W}(g(\x,\bu)) - \gamma_i W(\x)$ is bounded above by \eqref{eq:decreaseW}. By assumption, \eqref{eq:decreaseW} holds on each vertex of $\sigma_i.$ By convexity of linear inequalities, it therefore holds across all $\x\in\sigma_i.$ Because \eqref{eq:decreaseW} is enforced across all $\sigma \in \Tcal \subseteq \Xcal,$ \eqref{eq:decreaseBF} from Lemma \ref{lem:invSetCond} holds for all $\x \in \Tcal.$ By Lemma \ref{lem:invSetCond}, $\Scal \subseteq \Tcal$ is an invariant set. Additionally, because $\Tcal \subseteq \Xcal$, $\Scal$ is constraint admissible and therefore, a safe set.
\end{proof}

\begin{remark}
    If the system is autonomous, Theorem \ref{thm:invSetConditions} can simply consider $\bu_i = 0$ for all $i \in \mathbb{Z}_0^N$. 
\end{remark}

Constraint \eqref{eq:decreaseW} enforces the \ac{bf} condition \eqref{eq:decreaseBF} of Lemma \ref{lem:invSetCond}, with an additional error term on each vertex of the simplex, which results in \eqref{eq:decreaseBF} holding for all $\x \in\Xcal$. This error term corresponds to those used in \ac{cpa} Lyapunov function literature \cite[Theorem 2.10]{giesl2014computation}. It is also reminiscent of methods that use sampled $\epsilon-$nets to validate constraint adherence. For $\epsilon$-net methods, the maximum distance between sampled points, along with the Lipschitz continuity of the barrier function and the dynamical system, are used to ensure that function conditions hold for all $\x \in \Xcal$ \citep{robey2020learning,jin2020neural}. In contrast, the triangulation, $\Tcal$, and the corresponding structure imposed on $W$ in Theorem \ref{thm:invSetConditions} allow local relationships between points to be considered (via \eqref{eq:c}). As a result, Theorem \ref{thm:invSetConditions} only requires dense sampling in regions where it needed to successfully impose constraints, while $\epsilon$-net methods require uniform dense sampling across $\Xcal$ -- inherently demanding more samples.

Conservatism is introduced in $\Scal$ through the incorporation of the error term in Constraint \eqref{eq:decreaseBF}. This creates a limitation, as Theorem \ref{thm:invSetConditions} is only capable of determining safe sets that are contractive at the boundary. Any set where a point on the boundary maps to the boundary again after application of \eqref{eq:dynSys} would immediately violate \eqref{eq:decreaseW}.

\subsection{Implementation via Iterative Convex Overbounding}\label{sec:implement}

An optimization problem is developed to learn the \ac{cpa} \ac{bf}, $W = \{ W_\x \}_{\x \in \mathbb{E}_{\Tcal}}$, the decision function and $\Gamma = \{\gamma_i\}_{i=1}^{m_\Tcal}$ of system, while maximizing the volume of the safe set. The resulting problem is nonconvex due to the bilinearity of \eqref{eq:decreaseW}. Therefore, we use \ac{ico}, a nonconvex optimization technique \citep{warner2017iterative}. \ac{ico} is often used in control design for problems with bilinear and polynomial design variables \citep{warner2017iterative,sebe2018sequential,de2000convexifying}. In \ac{ico}, a design variable, $\X$, is decomposed to $\underline{\X} + \delta \X,$ where $\underline{\X}$ is constant and $\delta \X$ is a design variable. 
Any non-convexity involving $\delta \X$ is bounded above by convex constraints. Then, the new convex optimization problem is iteratively solved -- updating $\X$ with the solved $\delta \X$ at each iteration. 
A benefit of \ac{ico} methods are their convergence guarantees. If an \ac{ico} problem begins with a feasible solution, then the objective function of the optimization problem will decrease monotonically each iteration -- converging to a local minimum \citep{warner2017iterative}.

Problem \ref{prob:icoW} describes a single iteration of \ac{ico} to solve conditions in Theorem \ref{thm:invSetConditions}. The constant values of the \ac{bf},  $\underline{\W}=\{\underline{W}_{\x}\}_{\x \in \mathbb{E}_{\Tcal}}$ and classifier function, $\underline{\Gamma} = \{ \underline{\gamma}_i \}_{i=1}^{m_\Tcal}$, are assumed to be known from an initialization or a previous iteration of \ac{ico}. The aim is to determine $\delta \W$ and $\delta \Gamma$.

Complexity is introduced in the \ac{ico} problem by the minimization in \eqref{eq:decreaseW}, which is used to select the best input sample for a given state. In Problem \ref{prob:icoW}, we assume a constant input across each simplex. We also assume there is some selection of inputs for each simplex $\{\bu_{\underline{k}}\}_{\underline{k}=1}^{m_\Tcal}$ available from an initialization or a previous \ac{ico} iteration. The goal of Problem \ref{prob:icoW} then is to choose new  inputs for each simplex, $\Xi$ that improves the \ac{bf} condition with respect to $\underline{\W},$ $\underline{\Gamma},$ and $\underline{\Xi},$ as seen in \eqref{eq:chooseK}. 

\begin{problem}\label{prob:icoW}
Let the data set $\Dcal = \{\x_{z},\{\bu_{z,k}, \x_{z,k}^+\}_{k=1}^M\}_{z=1}^N$ be $NM$ one-step trajectories generated by \eqref{eq:dynSys} in $(\Xcal,\Ucal).$  Let $\Gamma = \underline{\Gamma} +\delta\Gamma = \{\gamma_i = \underline{\gamma}_i + \delta \gamma_i\}_{i=1}^{m_\Tcal} \subset \IR$ and $\Theta = \{\theta_\x\}_{\x \in \mathbb{E}_{\Tcal}}\subset \mathbb{R}.$ Let $\W = \underline{\W} +\delta \W = \{ W_\x = \underline{W}_{\x} +\delta W_\x\}_{\x \in \mathbb{E}_{\Tcal}}\subset \mathbb{R},$ where $W:\map{n}{}$ is the \ac{cpa} interpolation of $\W$ on $\Tcal$ for any $\x\in\sigma_i$. Define the $\bar{W}$ as in \eqref{eq:Wbar}. Finally, define $\Xi =\{\xi_i\}_{i=1}^{m_\Tcal} \subset \mathbb{Z}_1^M$ and $\underline{\Xi} =\{\underline{\xi}_i\}_{i=1}^{m_\Tcal} \subset \mathbb{Z}_1^M$, where \eqref{eq:icoAll} is satisfied for $\underline{\Xi}$, $\underline{\W}$, and $\underline{\Gamma}.$ The convex optimization problem is defined as
\begin{equation}
    \min_{\delta \W,\delta \Gamma, b, \Theta} C \nonumber
\end{equation}
\begin{subequations}\label{eq:icoAll}
\begin{align}
        \bar{W}(\x)\geq  \epsilon, \quad \quad &\forall \x^+ \notin \Tcal, \label{eq:wLabel_ico}\\
        \bar{W}(\x) \geq -\rho, \quad \quad &\forall \x \in \mathbb{E}_{\Tcal}, \\
        |\nabla \bar{W}_i| \leq b, \quad \quad &\forall i \in \mathbb{Z}_1^{m_{\Tcal}}, \label{eq:wGrad_ico} \\
        \M(\x_{i,j,\xi_i}^+, \x_{i,j}) \preceq \theta_{\x_{i,j}}\I, \label{eq:newDecrease} \quad \quad &\forall i \in \mathbb{Z}_1^{m_{\Tcal}}, \forall j \in \mathbb{Z}_0^n, \\
        \gamma_i \geq 0, \quad \quad &\forall i \in \mathbb{Z}_1^{m_\Tcal} \label{eq:alpha_ico} \\
        \theta_\x \geq 0, \quad \quad &\forall \x \in \mathbb{E}_{\Tcal}, \label{eq:theta_ico}
\end{align}
\end{subequations}
where
\begin{flalign}\label{eq:decreaseW_overbound}
    &  \M(\x_{i,j,\xi_i}^+, \x_{i,j}) = \nonumber \\ & \bmat{\bar{{W}}(\x^+_{i,j,\xi_i})-\underline{\gamma}_i \underline{W}(\x_{i,j}) - \delta \gamma_i \underline{W}(\x_{i,j})- \underline{\gamma}_i\delta  W(\x_{i,j})+ \underline{b} L_g c_{i,k} & \delta \gamma_i & \delta   W(\x_{i,j}) \\ * & -2 & 0 \\ * & * &-2},
\end{flalign}

\begin{flalign}\label{eq:chooseK}
    \xi_i = \argmin_{\xi\in\mathbb{Z}_1^M}  \underline{\bar{W}}(\x^+_{i,j,\xi})  \leq \underline{\bar{W}}(\x^+_{i,j,\underline{\xi}_i}) \quad \forall j \in \mathbb{Z}_0^n,
\end{flalign}
 $C$ is some convex cost function, $c_{i} = \max_{r,s\in \mathbb{Z}_0^n} \norm{\x_{i,r}-\x_{i,s}}_2$, and $\epsilon,\rho \in \IR_{>0}$ are equivalent to the values in Theorem \ref{thm:invSetConditions}.
\end{problem}

\begin{remark}
    For practical implementation, we use  $\xi_i = \argmin_{\xi\in\mathbb{Z}_1^M} \bigl(\max_{j\in\mathbb{Z}_0^n} \underline{{W}}(\x^+_{i,j,\xi}) - \underline{\gamma}_i\underline{W}(\x_{i,j}) + b L_g c_{i}\bigr)$ rather than \eqref{eq:chooseK}.
\end{remark}

In Problem \ref{prob:icoW}, the matrix \eqref{eq:decreaseW_overbound} bounds the barrier condition \eqref{eq:decreaseW} above (see Section \ref{sec:lemProof}).  Here, $M(\x_{i,j}, \x_{i,j}^+)$ is bounded above by the slack variable, $\theta(\x_{i,j})$, rather than $0$. The slack variables relax the barrier condition, allowing for a feasible initialization is found by iteratively minimizing $\Theta= \{\theta_\x\}_{\x \in \mathbb{E}_{\Tcal}}.$  
Lemma \ref{lem:icoSoln} states the conditions for which the solution to Problem \ref{prob:icoW} produces a valid barrier function and safe set. The proof of Lemma \ref{lem:icoSoln} can be found in Section \ref{sec:lemProof}.

\begin{lemma}\label{lem:icoSoln}
    The solution of Problem \ref{prob:icoW} adheres to the conditions of Theorem \ref{thm:invSetConditions} if $\nexists \x \in \mathbb{E}_{\Tcal}$ such that $\theta_\x > 0$. Therefore, the \ac{cpa} function, $W$, that satisfies Problem \ref{prob:icoW} is a barrier function and $\Scal \defeq \{\x\vecdim{n} \mid W(\x) \leq 0\}$ is a valid safe set.
\end{lemma}

In Algorithm \ref{alg:ico}, Problem \ref{prob:icoW} is iteratively solved to determine a valid barrier function, $W,$ decision function, $\Gamma$, and safe set, $\Scal$. The cost function of Problem \ref{prob:icoW} aims to minimize the slack variables, $\Theta$. Problem \ref{prob:icoW} is iteratively solved until all slack variables are zero and, as a result, a valid \ac{bf} and safe set have been determined. Then, the cost function changes to minimizing all positive values of $\W$, increasing the area of the safe set, $\Scal$.

\noindent
\begin{minipage}[t]{0.5\textwidth}
  \begin{cor}\label{cor:safeSet}
    If a valid initial barrier function is found by lines 2-7 of Algorithm \ref{alg:ico}, then Algorithm \ref{alg:ico} will produce a valid barrier function, $W$, and the safe set, $\Scal$.
\end{cor}
When a valid initial barrier function can be found, Algorithm \ref{alg:ico} will always produce a true safe set, $\Scal$.  However, note that the safe set $\Scal$ can be the null set if $W(\x) > 0$ for all $\x \in \Xcal$. Sections \ref{sec:convergeICO} and \ref{sec:additionalData} detail the convergence properties of Algorithm \ref{alg:ico} and the process of adding additional data during learning in Algorithm \ref{alg:ico}. 

\subsection{Computation Complexity}

Modeling the \ac{bf} as a \ac{cpa} function allows us to exploit local properties of the function and ensure that the \ac{bf} conditions hold for all $\x \in\Xcal$. These benefits come at the cost of computation.

First, we require a triangulation of the space over the vertex points created by the sampled points $\{\x_z\}_{z=1}^N.$ Forming Delaunay triangulations becomes computationally difficult beyond 6 dimensions \citep{hornus2008efficient}, but 

\end{minipage}\hfill
\begin{minipage}[t]{0.45\textwidth}
  \vspace{-3.5mm}  
  \begin{algorithm2e}[H]
    \caption{\ac{ico} for \ac{bf} synthesis}\label{alg:ico}
    \KwIn{$\Dcal$}
    Initialize $\W$ and $\Gamma$ (Section \ref{sec:icoInit})\;
    Set $C = \sum_{\x \in \mathbb{E}_\Tcal}\theta_\x$ in Problem \ref{prob:icoW}\;
    $\Delta C = \infty$\;
    \While{$\exists \x \in \mathbb{E}_{\Tcal}$ s.t. $\theta_\x > 0$ and $\Delta C \geq \chi$}{
     Solve Problem \ref{prob:icoW}\;
     $\underline{W} = \underline{W} + \delta W$, $\underline{\Gamma} = \underline{\Gamma} + \delta \Gamma$\;
     $\Delta C = C - C_{prev}$\;
    }
    \If{$\nexists \x \in \mathbb{E}_{\Tcal}$ s.t. $\theta_\x > 0$}{
    Set $C = \sum_{\x \in \mathbb{E}_\Tcal} \max(0,\W_\x)$, $\Delta C = \infty$\;
    Set $\Theta = 0$\;
    }
    \While{$\Delta C \geq \chi$}{
    Solve Problem \ref{prob:icoW}\;
    $\underline{W} = \underline{W} + \delta W$, $\underline{\Gamma} = \underline{\Gamma} + \delta \Gamma$\;
    $\Delta C = C - C_{prev}$\;
    }
    \KwResult{$W,\Gamma,\Scal$}
  \end{algorithm2e}
\end{minipage}

\noindent incremental algorithms can be used instead. Recent triangulation methods can handle up to 17 dimensions \citep{yadav2024resource}. Problem \ref{prob:icoW} requires solving $2N +m_\Tcal + 2m_\Tcal n$ linear constraints and $m_\Tcal(n+1)$ LMI constraints with matrix dimension $3\times 3$ (this dimension is invariant to state dimension, etc) each iteration. We use interior point methods to solve the series of convex optimization problems that result from \ac{ico}, which has polynomial complexity \citep{vandenberghe1996semidefinite}.

\section{Numerical Experiments}\label{sec:numEx}

Algorithm \ref{alg:ico} was tested on a nonlinear autonomous system and a linear non-autonomous system. In Figures \ref{fig:nonlinSys} and \ref{fig:linearSys_nonAut}, respectively, the boundary of the safe set, $W=0$, found by Algorithm \ref{alg:ico} is denoted by the light blue line. Both systems are then compared to a model based method of determining safe sets. In \ref{fig:nonlinSys}, the safe set of the nonlinear system is compared to the largest level set of the Lyapunov function, $V = (x^{(1)})^2+ (x^{(2)})^2$. Here, the safe set found by Algorithm \ref{alg:ico} actually has a larger volume than this model-based method. In \ref{fig:linearSys_nonAut}, the safe set found is compared to the maximal \ac{ci} set of the system, which is determined by model-based geometric methods \citep{MPT3}. Here, the volume of our safe set approaches that of the maximal \ac{ci} set.

The value of $\gamma$ across on each simplex of the triangulation, $\Tcal \subseteq \Xcal$, is also shown for each system in Figures \ref{fig:nonlinSys_gamma} and \ref{fig:linearSys_nonAut_gamma}, where darker colors show that $\gamma$ is near $0$ and bright yellow denotes high values. In \ref{fig:nonlinSys_gamma} in particular, it can be seen than $\gamma$ approximates the proximal function \eqref{eq:proxGamma}.

For more details of the numerical experiments see Sections \ref{sec:autoSys} and \ref{sec:nonAut}. Additionally, results from a linear, autonomous system are shown in \ref{sec:resLin} and results from implementing an adaptive sampling strategy are shown in \ref{sec:sampleRes}.


\begin{figure}[htbp]
\floatconts
{fig:results}
{\caption{The resulting safe sets found for the nonlinear autonomous system \ref{fig:nonlinSys} and the linear non-autonomous system \ref{fig:linearSys_nonAut}. The safe set boundary is in blue, while safe sets found using model based methods are shown as green polygons. The values of $\gamma$ across the triangulation $\Tcal \subseteq \Xcal$ are shown for the nonlinear autonomous system \ref{fig:nonlinSys_gamma} and linear non-autonomous system \ref{fig:linearSys_nonAut_gamma}. Black denotes a minimal value of $\gamma$ (near $0.1$) and bright yellow denotes the maximal value of $\gamma$ (ranging from 12.4 (\ref{fig:nonlinSys_gamma}) to 4.456 (\ref{fig:linearSys_nonAut_gamma})).}}
{
\subfigure[Safe set of nonlinear autonomous system.]{%
\label{fig:nonlinSys}
\includegraphics[width = 0.465\textwidth]{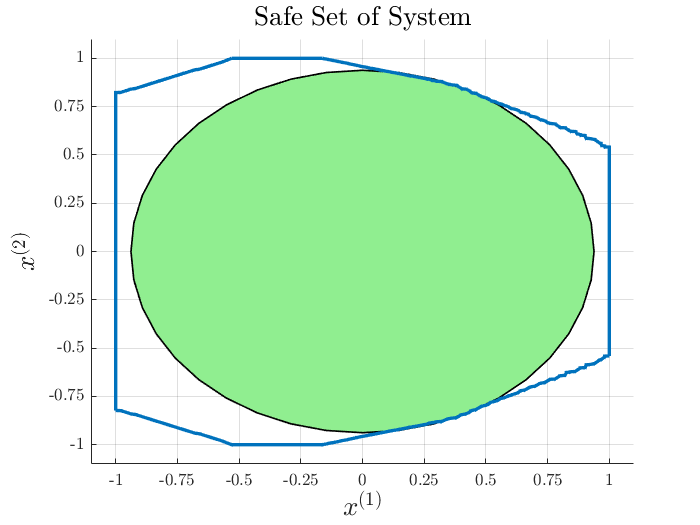}
}\qquad 
\subfigure[Values of $\gamma$ across $\Xcal$ for the nonlinear, autonomous system.]{%
\label{fig:nonlinSys_gamma}
\includegraphics[width = 0.465\textwidth]{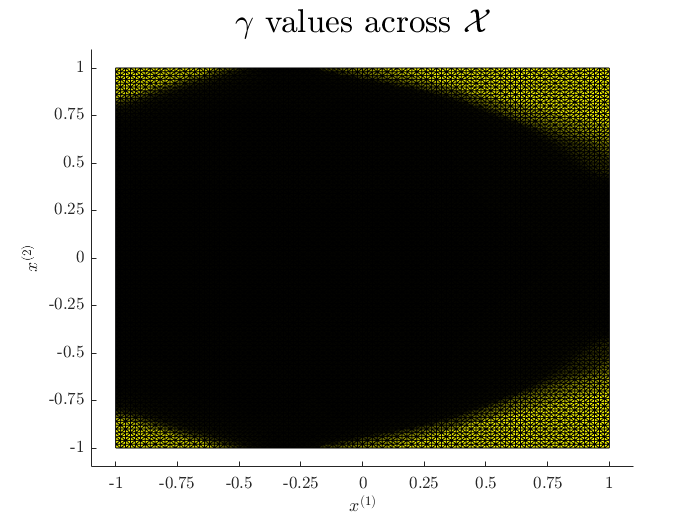}
}
\subfigure[Safe set of linear non-autonomous system.]{%
\label{fig:linearSys_nonAut}
\includegraphics[width = 0.465\textwidth]{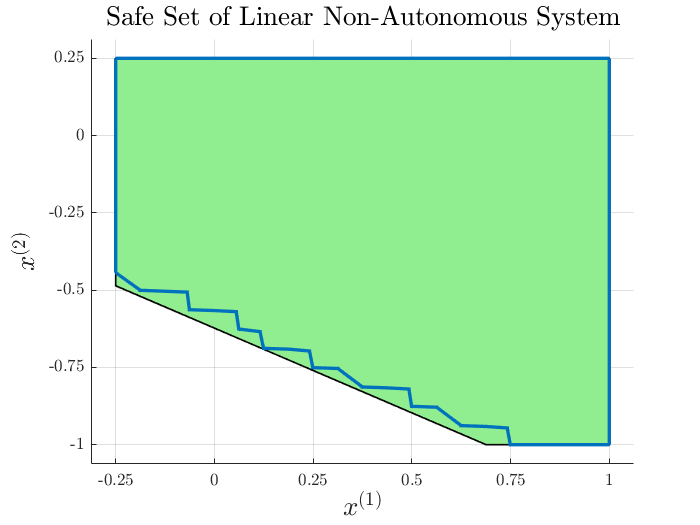}
}\qquad 
\subfigure[Values of $\gamma$ across $\Xcal$ for the linear, non-autonomous system.]{%
\label{fig:linearSys_nonAut_gamma}
\includegraphics[width = 0.465\textwidth]{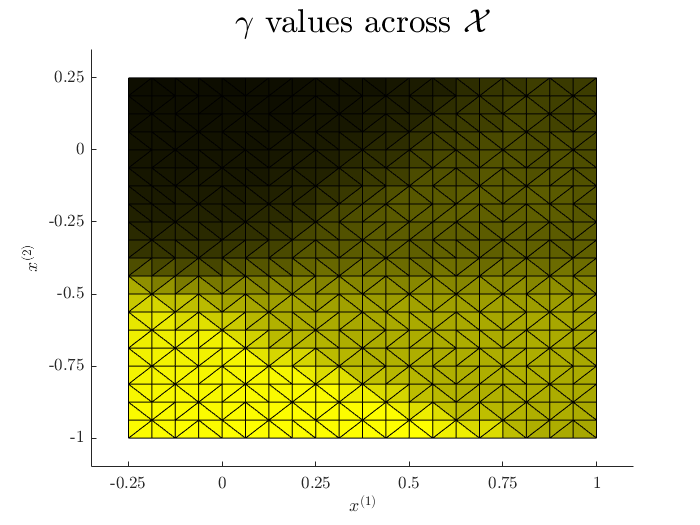}
}
}
\end{figure}




\acks{This work was supported by NSF Grant No. 2303158.}

\section*{References}

\bibliography{biblio}

\appendix

\section{Iterative Convex Overbounding}

\subsection{Proof of Lemma \ref{lem:icoSoln}}\label{sec:lemProof}

\begin{proof}
    In Problem \ref{prob:icoW}, Constraints \eqref{eq:wLabel_ico}-\eqref{eq:wGrad_ico} and \eqref{eq:alpha_ico} are equivalent to Conditions \eqref{eq:wLabel}-\eqref{eq:gradW} and \eqref{eq:boundAlpha} in Theorem \ref{thm:invSetConditions}. When $W(\x_{i,j}) = \underline{W}(\x_{i,j,\xi_i}) + \delta W(\x_{i,j,\xi_i})$ and $\gamma_i = \underline{\gamma}_i + \delta \gamma_i$, Condition \eqref{eq:decreaseW} can be written as
    \begin{flalign*}
        \bar{{W}}(\x^+_{i,j,\xi_i})-\underline{\gamma}_i \underline{W}(\x_{i,j}) - \delta \gamma_i \underline{W}(\x_{i,j})- \underline{\gamma}_i\delta  W(\x_{i,j}) - \delta\gamma_i\delta  W(\x_{i,j})+ b L_g c_i \leq 0,
    \end{flalign*}
    where $\delta\gamma_i\delta  W(\x_{i,j})$ are bilinear design variables. 
    By completion of the square, $\frac{1}{2}(\delta\gamma_i^2 + \delta  W(\x_{i,j}^2)) \geq - \delta\gamma_i\delta  W(\x_{i,j}).$ Therefore, the above expression can be bounded above by
    \begin{flalign*}
        \bar{{W}}(\x^+_{i,j,\xi_i})-\underline{\gamma}_i \underline{W}(\x_{i,j}) - \delta \gamma_i \underline{W}(\x_{i,j})- \underline{\gamma}_i\delta  W(\x_{i,j}) + \frac{1}{2}(\delta\gamma_i^2 + \delta  W(\x_{i,j}^2)) + b L_g c_i.
    \end{flalign*}
    If the above equation is less than 0 (i.e. $\theta_\x = 0$), then Condition \eqref{eq:decreaseW} also holds. This inequality can be equivalently expressed as \eqref{eq:decreaseW_overbound} by applying a Schur complement to the terms $\frac{1}{2}\delta\gamma_i^2$ and $\frac{1}{2}\delta W(\x_{i,j})^2.$ Therefore, if $\W$ and $\Gamma$ are solutions to Problem \ref{prob:icoW}, then they also satisfy Theorem \ref{thm:invSetConditions}.
\end{proof}

\subsection{Initialization of ICO}\label{sec:icoInit}

For nonconvex optimization, convergence is only guaranteed to local optima. Thus, initialization can drastically affect results.  Incorporating slack variables into Problem \ref{prob:icoW} allows for a smart initialization of the \ac{ico} problem.
 
\textit{Initialization:} All vertices $\x \in \mathbb{E}_\Tcal$ that have a corresponding $\x^+ \in \Tcal$ are assigned $W_\x = -\rho$ for the \ac{cpa} function, while all $\x \in \mathbb{E}_\Tcal$ with corresponding $\x^+ \notin \Tcal$ are assigned $W_\x \geq \epsilon.$ The values of $\Gamma$ are assigned $\gamma_i = 0.1$ when all vertices in a simplex $\sigma_i$ have the value $W_\x = -\rho$. Otherwise, $\gamma_i$ is assigned a value of $1.$  This initialization prioritizes having the largest area $\Scal$ possible.

\subsection{Convergence properties of Algorithm \ref{alg:ico}}\label{sec:convergeICO}

Theorem \ref{thm:icoConverge} then shows that the two optimization problems in Algorithm \ref{alg:ico} will always converge to a local minima. 
\begin{theorem}\label{thm:icoConverge}
    Algorithm \ref{alg:ico} converges to a local minimum of Problem \ref{prob:icoW} with $C = \sum_{\x \in \mathbb{E}_\Tcal}$ and Problem \ref{prob:icoW} with $C  = \sum_{\x \in \mathbb{E}_\Tcal} \max(0,W_\x)$. 
\end{theorem}
\begin{proof}
    For the first iteration of lines 3-7, Problem \ref{prob:icoW} will produce a feasible solution by the following logic. The triangulation $\Tcal \subseteq \Xcal$ and the function $W$ is bounded, producing a finite $\nabla W_i$ at each simplex $\sigma_i \in \Tcal.$ Thus, the design variable $b$ will be finite. This fact, along with Assumption \ref{assum:Lipschitz}, the constraint that $A \geq 0$, and the definition of $\bar{W}$ in \eqref{eq:Wbar}, ensure that the matrix $M(\x_{i,j}^+,\x_{i,j,\xi_i})$ is always bounded above and can therefore be bounded above with some finite matrix, $\theta(\x_{i,j})\I.$ 

    From there on, each iteration of lines 3-7 will inherit the feasibility of the previous problem as $\underline{W}+\delta {W}$ and $\underline{\Gamma} +\delta \Gamma$ are always feasible when $\delta W$ and $\delta \Gamma$ are zeros. Therefore, the cost function either decreases or remains the same. The cost function is bounded below (by \eqref{eq:theta_ico}), so the cost function will converge to a local minima. 

    When expanding the safe set of the barrier function (lines 8-16), the if statement on line 8 ensures that the initial $\underline{W}$ and $\underline{\Gamma}$ are feasible solutions to Problem \ref{prob:icoW} with $\Theta = 0.$ Once again, the new cost function is bounded by the bounds place on $\W$ in Problem \ref{prob:icoW}. By the logic above, this section of the algorithm will also converge to a local optima.
\end{proof}

\subsection{Searching for additional data}\label{sec:additionalData}

As the \ac{bf} and safe set are searched, we can add additional data points to facilitate either a feasible or improved solution. For initialization in \ac{ico} (lines 2-7) of Algorithm \ref{alg:ico}, this can be achieved by determining the simplex with the largest value of summed slack variables and adding a point to bisect the simplex using the longest edge bisection algorithm \citep{strong2025adaptive}.
When a valid barrier function is found, we can also add data points to the boundary of the barrier function -- bisecting all simplices where $W$ passes through $0.$ In this way, our iterative algorithm selects regions of the state space that require denser sampling.

To ensure feasibility of the solution, we add points so that the bisected simplices maintain the current barrier function. The added point will have a $W$ value that is the midpoint between the two existing vertex points, and the new simplices will have an $\gamma$ value than maintains the $\gamma$ of the originator simplex.

\section{Autonomous Systems: Numerical Results}\label{sec:autoSys}

\subsection{Results}\label{sec:resLin}

In addition to the results shown in Section \ref{sec:numEx}, we also tested Algorithm \ref{alg:ico} on a linear, autonomous system (described in Section \ref{sec:expDetails}). 

Figure \ref{fig:linearSys} shows the safe set found by Algorithm \ref{alg:ico} in blue -- compared to the maximal invariant set found by model based geometric methods \citep{MPT3} (shown as a green polygon). Figure \ref{fig:linearSysAlpha} shows the values of $\gamma$ over each simplex of the triangulation, $\Tcal \subseteq \Xcal,$ where a darker value indicates $\gamma$ is closer to 0.

Figure \ref{fig:bf_aut} shows the \ac{cpa} \ac{bf} found for the linear and nonlinear autonomous. Each simplex of the triangulation $\Tcal\subseteq \Xcal$ is mapped to an affine function.

\begin{figure}[htbp]
\floatconts
{fig:linEx}
{\caption{The resulting safe sets found for the linear autonomous system \ref{fig:linearSys} The safe set boundary is in blue, while the maximal safe set for the system as found by \citep{MPT3} is in green. In \ref{fig:linearSysAlpha}, the value of $\gamma$ are shown across the simplices of the triangualtion.}}
{
\subfigure[Linear autonomous system.]{%
\label{fig:linearSys}
\includegraphics[width = 0.465\textwidth]{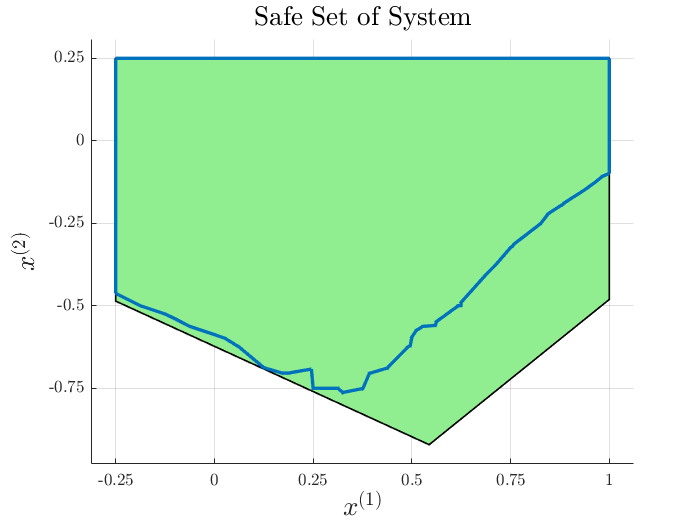}
}\qquad 
\subfigure[Value of $\gamma$ across $\Xcal$.]{%
\label{fig:linearSysAlpha}
\includegraphics[width = 0.465\textwidth]{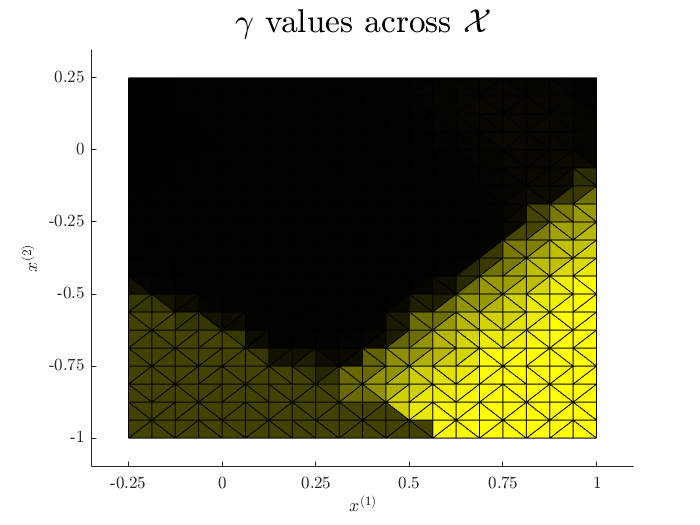}
}
}
\end{figure}

\begin{figure}[htbp]
\floatconts
{fig:bf_aut}
{\caption{The resulting CPA  barrier function found for the linear autonomous system \ref{fig:linearSys_bf} and nonlinear autonomous system \ref{fig:nonlinearSys_bf} are shown. The high data density for the nonlinear system results in very small simplices, while the linear system's \ac{bf} was generated with much fewer data points.}}
{
\subfigure[Linear autonomous system.]{%
\label{fig:linearSys_bf}
\includegraphics[width = 0.465\textwidth]{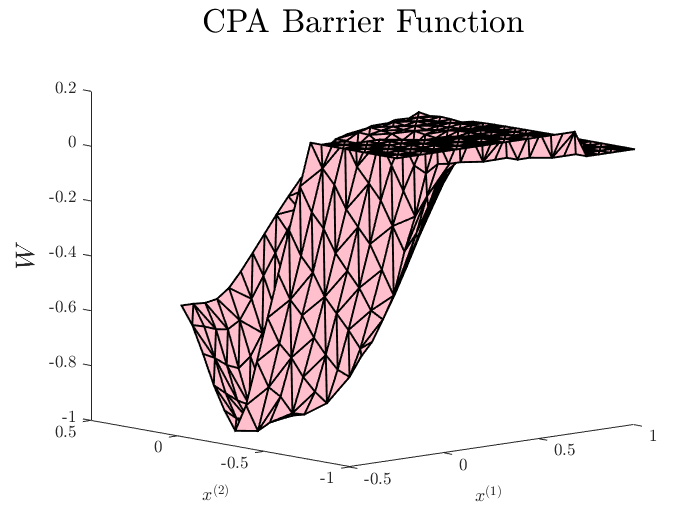}
}\qquad 
\subfigure[Nonlinear autonomous system.]{%
\label{fig:nonlinearSys_bf}
\includegraphics[width = 0.465\textwidth]{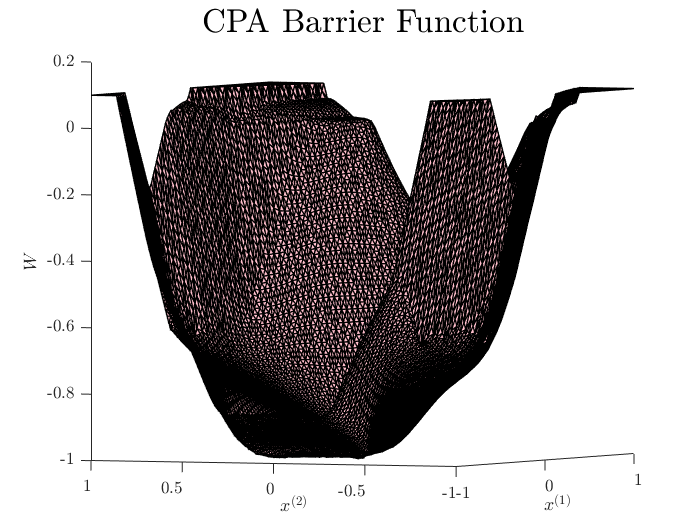}
}
}
\end{figure}

\subsection{Experimental details: Autonomous Systems}\label{sec:expDetails}

Algorithm \ref{alg:ico} was tested on two autonomous dynamical systems: a linear system and a nonlinear system. The dynamical systems and related experimental details are described below.
\begin{itemize}
    \item The autonomous linear system is described by the equation of motion
    \begin{equation}\label{eq:autLin}
        \x^+ = \bmat{0.22 & 0.4013\\ -0.5364 & 0.2109}\x.
    \end{equation}
    The Lipschitz constant of the system, $L$, is $0.5837$. The state constraint set considered was $\Xcal = [-0.25,1] \times [-1, 0.25]$. 
    
    In order to produce the safe set in Figure \ref{fig:linearSys} (and corresponding $\gamma$ values in Figure \ref{fig:linearSysAlpha}), a \ac{bf} was synthesized using grid sampled one-step trajectories in $\Xcal.$ A spacing of $0.0625$ was used between each point, producing a data set with $441$ one-step trajectories, i.e. $\{\x_z,\x_z^+\}_{i=1}^{441}.$ Algorithm \ref{alg:ico} found a feasible \ac{bf} at the 34th iteration. Algorithm \ref{alg:ico} then expanded that initial invariant set for an additional $9,966$ iterations with the final change in the cost between iteration being $1.7\times 10^{-5}$.
    \item The autonomous nonlinear system is described by
    \begin{equation}
        \x^+ = \bmat{0.5x^{(1)} -0.7(x^{(2)})^2 \\ 0.9(x^{(2)})^3 +x^{(1)}x^{(2)}}.
    \end{equation}
    The Lipschitz constant of the system used was $L  = 4.05$ and the state constraint set considered was $\Xcal = [-1, 1] \times [-1, 1]$.

    To produce results in Figures \ref{fig:nonlinSys} and \ref{fig:nonlinSys_gamma}, Algorithm \ref{alg:ico} was applied to a data set of one-step trajectories. This data set was created by sampling one-step trajectories over a grid of $0.02$ intervals covering $\Xcal$. This produced a total of $10,201$ one-step trajectories in the data set. Algorithm \ref{alg:ico} required 12 iterations to determine a feasible \ac{bf} and then iterated an additional $66$ iterations to expand the set with the final change in the cost between iterations being $0.0597.$

\end{itemize}

Algorithm \ref{alg:ico} requires convex optimization, which is solved using Mosek version 9.3.13 \citep{mosek} and YALMIP \citep{Lofberg2004} in MATLAB. Solving Problem \ref{prob:icoW} (a single iteration of Algorithm \ref{alg:ico}) typically took between 1 and 110 seconds (depending on the data set size) with a maximum number of computational threads being 8.

\subsection{Sampling during BF synthesis}\label{sec:sampleRes}

Section \ref{sec:additionalData} provides guidelines for sampling during \ac{bf} synthesis. Here, we provide brief results for the linear dynamical system \eqref{eq:autLin}.

The initial data set used to synthesize a \ac{bf} was a grid sampling of one step trajectories over $\Xcal = [-0.25,1] \times [-1, 0.25]$, where the sampling interval was $0.25$. Additional data points were added when finding a feasible \ac{bf} by determining the simplex with the highest summed slack variable and then adding an additional data point to bisect the simplex. During safe set expansion, additional data points were added along the current boundary of the safe set. Overall, $453$ one step trajectores were used.

Figure \ref{fig:addPoints} shows the safe set found in \ref{fig:linearSys_addPoints} and the corresponding values of $\gamma$ on each simplex of the triangulation in \ref{fig:linearSys_addPoints_alpha}. Note that the triangulation is no longer uniform across $\Xcal.$ Instead, refinements are concentrated in certain regions -- typically near the border of the safe set.

\begin{figure}[htbp]
\floatconts
{fig:addPoints}
{\caption{In \ref{fig:linearSys_addPoints},  the safe set found for the linear autonomous system when using adaptive sampling during Algorithm \ref{alg:ico} is shown. The safe set boundary is in blue, which approaches the maximal safe set of the system (find via model based geometric methods \citep{MPT3}), which is shown as the green polygon. In \ref{fig:linearSys_addPoints_alpha}, the corresponding $\gamma$ values across the triangulation are shown -- with black representing a smaller value of $\gamma$ near the minimum of 0.01 and yellow representing a higher value of $\gamma$ near the maximum of 14.45.}}
{
\subfigure[Linear autonomous system.]{%
\label{fig:linearSys_addPoints}
\includegraphics[width = 0.465\textwidth]{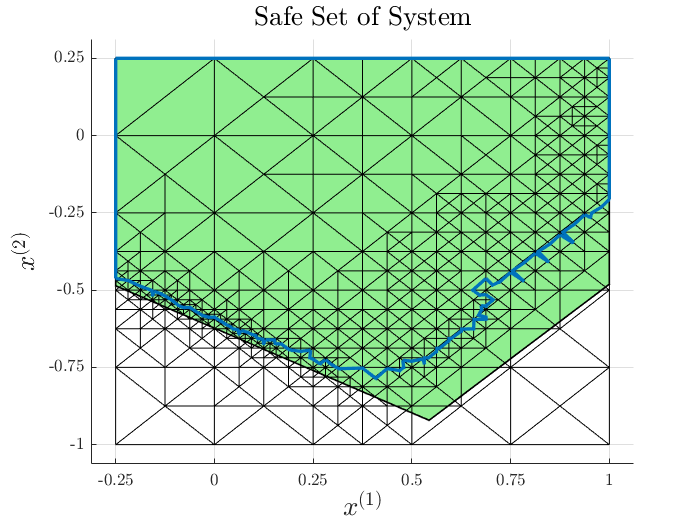}
}\qquad 
\subfigure[$\gamma$ values across the triangulation.]{%
\label{fig:linearSys_addPoints_alpha}
\includegraphics[width = 0.465\textwidth]{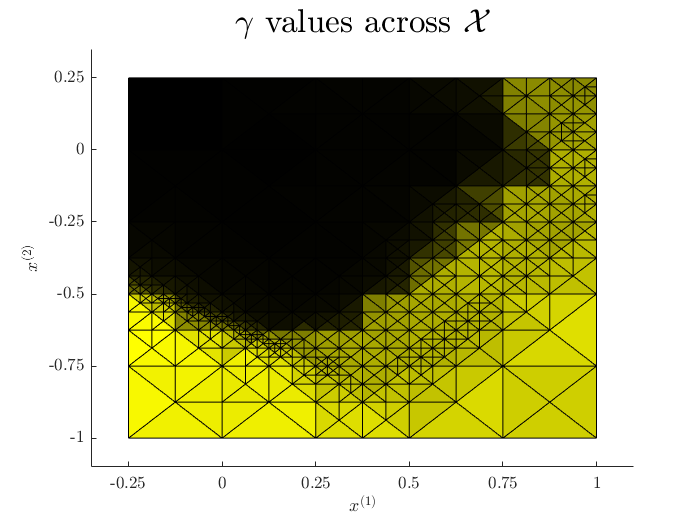}
}
}
\end{figure}

 \pagebreak
\section{Non-Autonomous Systems: Numerical Results}\label{sec:nonAut}


Algorithm \ref{alg:ico} was tested on a linear, non-autonomous dynamical system  described by the equation of motion
    \begin{equation}\label{eq:nonAutLin}
        \x^+ = \bmat{0.22 & 0.4013\\ -0.5364 & 0.2109}\x + \bmat{0\\1} \bu.
    \end{equation}
    The state constraint set considered was $\Xcal = [-0.25,1] \times [-1, 0.25]$ and $\Ucal = [-1,1].$ 
    The Lipschitz constant of the system with respect to $\x$ is $0.5837$, and the Lipschitz constant with respect to $\bu$ is $1$. Because the $\B$ matrix is constant, the Lipschitz continuity of the system can be bounded as $\norm{\A\x+\B\bu - \A\y - \B\w} \leq \norm{\A}\norm{\x-\y} +\norm{\B}\norm{\bu-\w}$ via application of the triangle inequality and the properties of matrix operator norms.
    
    The \ac{bf} was synthesized using grid sampled one-step trajectories in $\Xcal \times \Ucal.$ A spacing of $0.0625$ was used between each state. The inputs were sampled from $-1$ to $1$ with an interval of $0.1$. This results in the data set producing a data set with $441\times 21$ one-step trajectories, i.e. $\Dcal = \{\x_z,\{\bu_{z,k},\x_{z,k}^+\}_{k=1}^{21}\}_{z=1}^{441}.$ 

    Algorithm \ref{alg:ico} found a feasible \ac{bf} at the 10th iteration. The safe set was then expanded for an additional 13 iterations, with minimal change in the cost function for the last 4 iterations.

The \ac{cpa} \ac{bf} found for the non-autonomous linear system is shown in Figure \ref{fig:nonAut_bf}.

\begin{figure}[htbp]
\floatconts
{fig:nonAut_bf}
{\caption{The \ac{cpa} \ac{bf} found for the linear, non-autonomous system is shown, where the function is affine on each simplex of the triangualtion.}}
{
\subfigure[Linear non-autonomous system.]{%
\label{fig:linNonAutCBF}
\includegraphics[width = 0.465\textwidth]{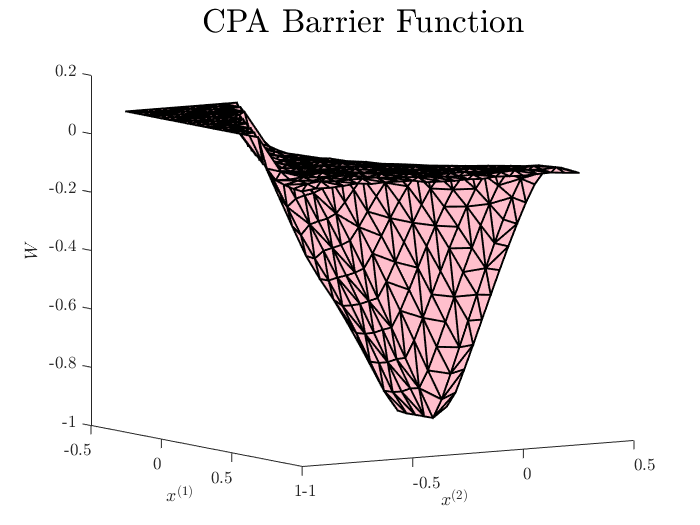}
}
}
\end{figure}

\end{document}